\documentclass[conference]{IEEEtran}
\IEEEoverridecommandlockouts

\usepackage{versions}
\includeversion{ARXIV}
\excludeversion{INFOCOM}
\makeatletter
\def\ps@headings{%
\def\@oddhead{\mbox{}\scriptsize\rightmark \hfil \thepage}%
\def\@evenhead{\scriptsize\thepage \hfil \leftmark\mbox{}}%
\def\@oddfoot{}%
\def\@evenfoot{}}
\makeatother
\usepackage{cite}
\usepackage[dvips]{graphicx}        % declare the path(s) where your graphic files are
	\graphicspath{{../figures/}}      % and their extensions so you won't have to specify
  \DeclareGraphicsExtensions{.eps}	% these with every instance of \includegraphics
\usepackage{amssymb}
\usepackage[cmex10]{amsmath}
\newtheorem{theorem}{Theorem}
\newtheorem{lemma}{Lemma}
\newtheorem{corollary}{Corollary}

\newtheorem{definition}{Definition}

\begin{document}
\title{Distributed Admission Control without\\Knowledge of the Capacity Region}
\author{
	\IEEEauthorblockN{Juan Jos\'e Jaramillo}
	% LONG AFFILIATION
	\IEEEauthorblockA{Dept. of Applied Math and Engineering\\
	Universidad EAFIT\\
	Medell\'in, Colombia\\
	E-mail: jjaram93@eafit.edu.co}
	% LONG AFFILIATION
	\and
	\IEEEauthorblockN{Lei Ying}
	% LONG AFFILIATION
	\IEEEauthorblockA{School of Electrical, Computer\\and Energy Engineering\\
	Arizona State University\\
	Tempe, AZ 85287 USA\\
	E-mail: lei.ying.2@asu.edu}
	% LONG AFFILIATION
	\thanks{Research supported in part by NSF Grants 1262329 and 1264012.}% <-this % stops a space
	% SHORT AFFILIATION
	%\thanks{J. J. Jaramillo is with the Department of Applied Math and Engineering, Universidad EAFIT, Medell\'in, Colombia (e-mail: jjaram93@eafit.edu.co).}% <-this % stops a space
	%\thanks{Lei Ying is with the School of Electrical, Computer and Energy Engineering, Arizona State University, Tempe, AZ 85287 USA (e-mail: lei.ying.2@asu.edu)}% <-this % stops a space
	% SHORT AFFILIATION
}
\maketitle

%%%%%%%%%%%%%%%%%%%%%%%%%%%%%%%%%%%%%%%%%%%%%%%%%%%%%%%%%%%%%%%%%%%%
%
% *** ABSTRACT ***
%
%%%%%%%%%%%%%%%%%%%%%%%%%%%%%%%%%%%%%%%%%%%%%%%%%%%%%%%%%%%%%%%%%%%%
\begin{abstract}

We consider the problem of distributed admission control without knowledge of the capacity region in single-hop wireless networks, for flows that require a pre-specified bandwidth from the network. We present an optimization framework that allows us to design a scheduler and resource allocator, and by properly choosing a suitable utility function in the resource allocator, we prove that existing flows can be served with a pre-specified bandwidth, while the link requesting admission can determine the largest rate that it can get such that it does not interfere with the allocation to the existing flows.

\end{abstract}

%%%%%%%%%%%%%%%%%%%%%%%%%%%%%%%%%%%%%%%%%%%%%%%%%%%%%%%%%%%%%%%%%%%%
%
% *** INTRODUCTION ***
%
%%%%%%%%%%%%%%%%%%%%%%%%%%%%%%%%%%%%%%%%%%%%%%%%%%%%%%%%%%%%%%%%%%%%
\section{Introduction}

Wireless networks are not only expected to transfer data, but also expected to provide multimedia services such as video and call conferencing. A common factor of such services is that they have Quality of Service (QoS) requirements. In this paper, we focus in the case that single-hop flows require a pre-specified bandwidth from the network, so we must design an algorithm that must determine if there are enough resources to fulfill a new request, given that the system is already serving a set of flows.

This is the problem of admission control, and it has been extensively studied for wireline networks. However, the case of admission control in wireless networks is more challenging, since the wireless channel is unreliable and susceptible to interference. Hence, the load imposed by a link that requests a pre-specified bandwidth from the network not only depends of the requested bandwidth, but also of the topology of the network due to contention among different transmitting links. Therefore, it is not obvious that the techniques developed for admission control in wireline networks can be used directly in wireless networks.

Several papers have highlighted the difficulties of guaranteeing QoS requirements in wireless networks. Specifically, \cite{Tang05} emphasizes the need to consider load balancing to maximize the number of admitted flows, while the importance of taking contention into account to determine the available bandwidth has been highlighted in \cite{Sanzgiri06, Chakeres07}. The problem is that determining the capacity region in wireless networks is not a trivial task, and the problem is further complicated if we want to implement admission control distributedly.

The main contributions of the paper are as follows:
\begin{enumerate}
	\item We present an utility maximization framework for resource allocation in single-hop wireless networks that allows us to design a distributed solution to the problem of admission control.

	\item By introducing a specific utility function, we prove that we can induce in the optimization framework an assignment of a pre-specified bandwidth to flows already admitted, while at the same time figuring out the maximum available bandwidth that can be assigned to the new flow.

	\item We present the conditions that the parameters in the utility function must fulfill to guarantee that the allocation in the stochastic system is asymptotically close to the desired allocation that has been induced in the static optimization framework.
\end{enumerate}

%%%%%%%%%%%%%%%%%%%%%%%%%%%%%%%%%%%%%%%%%%%%%%%%%%%%%%%%%%%%%%%%%%%%
%
% *** RELATED WORK ***
%
%%%%%%%%%%%%%%%%%%%%%%%%%%%%%%%%%%%%%%%%%%%%%%%%%%%%%%%%%%%%%%%%%%%%
\begin{ARXIV}
\section{Related Work}
\label{related_work}

Estimating the available bandwidth to do admission control has been an active topic of research. In \cite{Yang05, Sanzgiri06}, the problem of determining the impact of contention to find the available bandwidth is studied for multi-hop wireless networks, while \cite{Chakeres07} studies the problem of bandwidth estimation at a node. In \cite{Xue03}, the problem of contention is taken into account under the implicit assumption that the interference and transmission range of a node are the same. Some heuristics to support QoS are presented in \cite{Lee00, Ahn02}, but contention is ignored during admission control.

The use of packet scheduling to guarantee QoS in multi-hop networks has been considered in \cite{Luo04}. Some solutions for admission control have be implemented centrally \cite{Ramanathan98, Shah05}, or assume a specific wireless technology like TDMA \cite{Liao02, Zhu02, Guo04} or CDMA over TDMA \cite{Lin99, Liu05}. Using implicit synchronization in CSMA/CA networks, a performance similar to TDM is achieved in \cite{Singh07}.

Under the assumption that requests can be split, \cite{Tang05} proposes a solution for multi-hop multichannel networks. If the requests cannot be split, a heuristic is presented using the least-congested, minimum-hop path to route requests.

The idea of achieving provably good performance without any assumptions on the model of the arrival requests was presented in \cite{Jaramillo08b}, using ideas first developed for wireline networks in \cite{Awerbuch93, Awerbuch94a, Plotkin95}. The idea is to study a worst-case model to determine the impact of imperfect decisions due to the lack of knowledge of future requests.

To the best of our knowledge, this is the first paper that proves that we can accurately estimate the available bandwidth during admission control in wireless networks in a distributed manner, and without the need to know or estimate the capacity region.
\end{ARXIV}

%%%%%%%%%%%%%%%%%%%%%%%%%%%%%%%%%%%%%%%%%%%%%%%%%%%%%%%%%%%%%%%%%%%%
%
% *** MODEL ***
%
%%%%%%%%%%%%%%%%%%%%%%%%%%%%%%%%%%%%%%%%%%%%%%%%%%%%%%%%%%%%%%%%%%%%
\section{System Model}
\label{system_model}

In this section we introduce the definitions and assumptions that we use to model our system.

We consider a network that is represented by a graph $\mathcal{G}=\{ \mathcal{N},\mathcal{D} \}$, where $\mathcal{N}$ is the set of nodes and $\mathcal{D}$ is the set of directional links such that for all $n_1,n_2 \in \mathcal{N}$, if $(n_1,n_2) \in \mathcal{D}$ then node $n_1$ can transmit to node $n_2$. To identify the links, we number them sequentially. Denote by $D=|\mathcal{D}|$ the number of links in the network, and by abusing notation, we sometimes use $l \in \cal D$ to mean $l\in\{1,2,\ldots, D\}$.

Consider that time is slotted, and assume for simplicity that all packets have fixed size such that one packet can be transmitted in a single time slot. We use a \emph{schedule} to denote the set of links that are allowed to transmit in a given time slot. A \emph{feasible schedule} $\mathbf{s}=\{ s_{l} \}_{l \in \mathcal{D}}$ is a $D$-dimensional vector that satisfies the following properties:
\begin{enumerate}
	\item $s_{l} \in \{ 0, 1 \}$ for all $l \in \mathcal{D}$, where $s_{l}=1$ means that link $l$ is scheduled to transmit in the current time slot, and $s_l=0$ otherwise. We assume that each link can transmit at most one packet per time slot.
	\item $\mathbf{s}$ must satisfy the interference constraints of the network. In other words, for any links $l_1, l_2 \in \mathcal{D}$, and any feasible schedule $\mathbf{s}$, if $s_{l1}=s_{l2}=1$, then links  $l_1$ and $l_2$ can transmit simultaneously without interfering with each other.
\end{enumerate}
We denote the set of all feasible schedules by $\mathcal{S}$.

We assume that we do not do channel estimation before a packet is transmitted, and we denote by $c_l$ the state of the channel in a given time slot at link $l \in \mathcal{D}$, where $c_l = 1$ means that the channel is ON and a packet transmission will be successful. Similarly, $c_l=0$ means that the channel is OFF. Assume that for all links, $c_l$ is a Bernoulli random variable with mean $\bar{c}_l>0$, that is independent across time slots, and we only get to know its actual value after attempting transmission. Denote by $\mathbf{c} = \{c_l \}_{l \in \mathcal{D}}$ the vector of channel states at a given time slot.

It must be noted that depending on the schedule and transmission parameters\footnote{E.g., modulation, power level, coding, etc.} that we use, link reliability may vary. For example, if you schedule only one link at any time, link reliability may be higher than trying to simultaneously schedule as many links as possible. Thus, we could extend the concept of the channel state to allow for the possibility that $\bar{c}_l$ could vary depending on the schedule or transmission parameters used, but this channel model, albeit more realistic, does not give us more insight into our problem. Hence, for ease of explanation, we will only consider in this paper a simple channel model.

The admission control problem that we are studying is the following: assuming that a subset $\mathcal{L} \subset \mathcal{D}$ can been served with mean flow rate $\bar{x}$, and that link $l \in \mathcal{D} \setminus \mathcal{L}$ requests to be admitted with flow rate  $\bar{x}$, can we determine, without knowledge of the capacity region and without disturbing the service rates in the set $\mathcal{L}$, if  link $l$ can be served?

%%%%%%%%%%%%%%%%%%%%%%%%%%%%%%%%%%%%%%%%%%%%%%%%%%%%%%%%%%%%%%%%%%%%
%
% *** OPTIMIZATION FRAMEWORK ***
%
%%%%%%%%%%%%%%%%%%%%%%%%%%%%%%%%%%%%%%%%%%%%%%%%%%%%%%%%%%%%%%%%%%%%
\section{Optimization Framework}
\label{optimization_framework}

Now, we formally present the utility maximization framework that we will use later to develop a suitable admission controller. To do that, we consider the problem when we only need to serve a given subset of links.

Consider a subset $\mathcal{L} \subseteq \mathcal{D}$ of links that will be served. Let $L=|\mathcal{L}|$ be the number of links on the set. Without loss of generality, assume that links are numbered $1$ through $L$, and as mentioned before for the set $\mathcal{D}$, we sometimes use $l \in \mathcal{L}$ to mean $l \in \{1, 2, \ldots , L\}$.

For this case, we limit the set of feasible schedules to those  such that $s_l = 0$ for all $l \in \mathcal{D} \setminus \mathcal{L}$; that is, the set of feasible schedules that only serve links in $\mathcal{L}$. We denote the restricted set of feasible schedules by $\mathcal{S(\mathcal{L})} \subseteq \mathcal{S}$. A \emph{scheduling policy over $\mathcal{L}$} is defined as a probability function $P_{\mathcal{L}}(\mathbf{s})$ that indicates the probability of using schedule $\mathbf{s} \in \mathcal{S}$ in a given time slot, such that
\begin{equation*}
	P_{\mathcal{L}}(\mathbf{s}) \geq 0 \text{ for } \mathbf{s} \in \mathcal{S}(\mathcal{L}),
\end{equation*}
\begin{equation*}
	P_{\mathcal{L}}(\mathbf{s}) = 0 \text{ for all } \mathbf{s} \in \mathcal{S} \setminus \mathcal{S(\mathcal{L})}, \text{ and}
\end{equation*}
\begin{equation*}
	\sum_{\mathbf{s} \in \mathcal{S}} P_{\mathcal{L}}(\mathbf{s}) = 1.
\end{equation*}
Observe that, from the definition of $P_{\mathcal{L}}(\mathbf{s})$, we have
\begin{equation*}
	\sum_{\mathbf{s} \in \mathcal{S}} P_{\mathcal{L}}(\mathbf{s}) = \sum_{\mathbf{s} \in \mathcal{S}(\mathcal{L})} P_{\mathcal{L}}(\mathbf{s}) = 1.
\end{equation*}

Noting that a transmission can only be successful if the channel is ON, i.e. $c_l=1$, we have that $c_l s_l$ denotes the number of successful transmission attempts at link $l$ in a given time slot. We assume that we can schedule a transmission even if there are no packets available, in which case a \emph{null packet} is transmitted. Therefore, the average service rate to link $l \in \mathcal{D}$ is bounded by
\begin{equation}
\label{def_mu_l}
	\mu_l \leq \sum_{\mathbf{s} \in \mathcal{S}} \sum_{c_l = 0}^1 c_l s_l P(c_l) P_{\mathcal{L}}(\mathbf{s}),
\end{equation}
where \eqref{def_mu_l} makes explicit the fact that the distribution of $c_l$ is independent on the schedule $\mathbf{s}$. It should be noted from the definition of $P_{\mathcal{L}}(\mathbf{s})$ that $\mu_l = 0$ for all $l \in \mathcal{D} \setminus \mathcal{L}$. Simplifying \eqref{def_mu_l} we get
\begin{equation}
\label{simplified_mu_l}
	\mu_l \leq \sum_{\mathbf{s} \in \mathcal{S}(\mathcal{L})} \bar{c}_l s_l P_{\mathcal{L}}(\mathbf{s}).
\end{equation}

Observe that if in a given time slot we use schedule $\mathbf{s}$, then on average we will have $\bar{c}_l s_l$ successful transmissions at link $l$.
\begin{definition}
\label{def_av_rates}
	The set $\Gamma(\mathcal{L})$ of average successful transmissions at any time slot is the set of vectors $\{ \bar{c}_l s_l \}_{l \in \mathcal{D}}$ where $\mathbf{s} \in \mathcal{S}(\mathcal{L})$.
$\hfill \diamond$
\end{definition}

\begin{definition}
\label{def_capacity}
	The capacity region $\mathcal{C}(\mathcal{L})$, restricted to the set $\mathcal{L}$, is the set of average service rates $\mu = \{ \mu_l \}_{l \in \mathcal{D}}$ such that there exists a scheduling policy $P_{\mathcal{L}}(\mathbf{s})$ and \eqref{simplified_mu_l} holds true for all $l \in \mathcal{D}$.
$\hfill \diamond$
\end{definition}

From Definitions \ref{def_av_rates} and \ref{def_capacity}, we observe that $\mathcal{C}(\mathcal{L})$ is the convex hull of $\Gamma(\mathcal{L})$.

Associated with every link, we define a utility function $U_l(x_l)$, that is a function of the mean assigned flow rate $x_l$. By properly choosing a suitable utility function, it can be seen that the following optimization problem can have different resource allocation solutions:
\begin{equation}
\label{opt_problem_old}
	\max_{ \mu \in \mathcal{C}(\mathcal{L}), \mathbf{x} } \sum_{l \in \mathcal{L}} U_l(x_l)
\end{equation}
subject to
\begin{align*}
	& 0 \leq x_l \leq \mu_l \text{ for all } l \in \mathcal{L} \\
	& x_l = 0 \text{ for all } l \in \mathcal{D} \setminus \mathcal{L}.
\end{align*}
We will denote by $(\mu^*, \mathbf{x}^*)$ a solution to \eqref{opt_problem_old}. Note that the solution may not be unique, but the optimal value is. In Section \ref{admission_control_problem} we will introduce a utility function that will allow us to allocate resources such that we can solve the admission control problem.

\section{The Admission Control Problem}
\label{admission_control_problem}

To solve this problem, we first present a scheduler and resource allocator in Section \ref{scheduler_ra}, and in Section \ref{convergence_analysis} we prove the convergence results that guarantee that the stochastic system is stable and that resources are assigned according to the utility functions that we define for each link. Later in Section \ref{solution_admission_control} we introduce a suitable utility function and prove that it forces the resource allocator to guarantee a mean assigned rate of $\bar{x}$ to all links in $\mathcal{L}$, while at the same time it accurately estimates the maximum available rate for the new link, allowing us to make an admission control decision.

%%%%%%%%%%%%%%%%%%%%%%%%%%%%%%%%%%%%%%%%%%%%%%%%%%%%%%%%%%%%%%%%%%%%
% *** SCHEDULER AND RESOURCE ALLOCATOR ***
%%%%%%%%%%%%%%%%%%%%%%%%%%%%%%%%%%%%%%%%%%%%%%%%%%%%%%%%%%%%%%%%%%%%
\subsection{Scheduler and Resource Allocator}
\label{scheduler_ra}
Using a dual decomposition approach similar to the one used in \cite{Jaramillo11b}, we propose the following scheduler for serving the set of links $\mathcal{L}$ at time slot $t$
\begin{equation}
\label{scheduler}
	\mathbf{s}(t) \in \mathop{\arg \max}_{ \mathbf{s} \in \mathcal{S}(\mathcal{L}) } \sum_{l \in \mathcal{L}} q_l(t) \bar{c}_l s_l
\end{equation}
and the distributed resource allocator\footnote{Note that on the literature it is also known as \emph{congestion controller}. Later we will prove that it allocates resources to links depending on the specified utility functions. Thus, in this paper we will prefer the term \emph{resource allocator} to make clear the fact that later we will choose the utility functions to force the system to allocate resources such that we can make an admission control decision.} at link $l \in \mathcal{L}$
\begin{equation}
\label{congestion_controller}
	x_l(t) \in \mathop{\arg \max}_{ 0 \leq x_l \leq X_{max} } \frac{1}{\epsilon} U_l(x_l) - q_l(t) x_l,
\end{equation}
where  $\epsilon>0$ is a fixed-size parameter, $X_{max}>0$ is a large enough parameter, and $\mathbf{q}(t) = \{ q_l(t) \}_{l \in \mathcal{L}}$ are the queue lengths at every link. We need to convert $\mathbf{x}(t) = \{ x_l(t) \}_{l \in \mathcal{L}}$, which may not necessarily be an integer, into the number of packets that are admitted into the network at time slot $t$, which we denote by $\mathbf{a}(t)$. To do the conversion, which can be done in many different ways, we specify that $a_l(t)$ is a random variable with mean $x_l(t)$ and finite variance upper-bounded by $\sigma^2$, such that $P(a_l(t)=0)>0$ and $P(a_l(t)=1)>0$. The last two conditions guarantee that the Markov chain $\mathbf{q}(t)$ is irreducible and aperiodic.

Denoting by
\begin{equation}
\label{def_dl}
d_l(t) = c_l(t) s_l(t)
\end{equation}
the number of successfully transmitted packets at link $l$ in time slot $t$, we note that $\mathbf{q}(t)$ is updated with the equations
\begin{equation*}
	q_l(t+1) = \left[ q_l(t) + a_l(t) - d_l(t) \right]^+
\end{equation*}
for all $l \in \mathcal{L}$, where for any $\alpha \in \mathbb{R}$, $\alpha^+ = \max \{ \alpha,0 \}$. 

%%%%%%%%%%%%%%%%%%%%%%%%%%%%%%%%%%%%%%%%%%%%%%%%%%%%%%%%%%%%%%%%%%%%
% *** CONVERGENCE ANALYSIS ***
%%%%%%%%%%%%%%%%%%%%%%%%%%%%%%%%%%%%%%%%%%%%%%%%%%%%%%%%%%%%%%%%%%%%
\subsection{Convergence Analysis}
\label{convergence_analysis}
We now proceed to present the convergence results that prove that \eqref{scheduler} and \eqref{congestion_controller} keep the queues stable and that allocate resources optimally.

\begin{lemma}
\label{queue_stability}
	Assume that $X_{max} \geq \max_{\mu \in \mathcal{C}(\mathcal{L})} \left\{ \max_{l \in \mathcal{L}} \{ \mu_l \} \right\}$. If there exists $\mu(\Delta) \in \mathcal{C}(\mathcal{L})/(1+\Delta)$ for some $\Delta > 0$ such that $\mu_l(\Delta) > 0$ for all $l \in \mathcal{L}$, then
\begin{equation*}
	\lim_{t \rightarrow \infty} E \left[ \sum_{l \in \mathcal{L}} q_l(t) \right] \leq B_3 + \frac{1}{\epsilon} B_4
\end{equation*}
given $\epsilon > 0$ and for some $B_3>0$, $B_4>0$.
$\hfill \diamond$
\end{lemma}

\begin{ARXIV}
The proof can be found in Appendix \ref{proof_queue_stability}.
\end{ARXIV}
\begin{INFOCOM}
The interested reader can see the proof of the Lemma in \cite{Jaramillo13}.
\end{INFOCOM}
Lemma \ref{queue_stability} tells us that, if there exists a resource allocation vector such that we can serve all the links with a non-zero mean rate, then our algorithm can stabilize the queues.

\begin{lemma}
\label{optimality_result}
	Given $\epsilon > 0$, and assuming that $U_l(\cdot)$ is a concave function, we have that
\begin{equation}
\label{optimality_equation}
	\lim_{T \rightarrow \infty} \sum_{l \in \mathcal{L}} \left\{ U_l( x_l^* ) - U_l\left( E\left[ \frac{1}{T} \sum_{t=1}^T  x_l(t) \right] \right) \right\} \leq B_1 \epsilon,
\end{equation}
for some $B_1>0$, where $\mathbf{x}^*$ a solution to \eqref{opt_problem_old} and  $\mathbf{x}(t)$ is a solution to \eqref{congestion_controller}.
$\hfill \diamond$
\end{lemma}

\begin{ARXIV}
In Appendix \ref{proof_optimality_result} the proof can be found.
\end{ARXIV}
\begin{INFOCOM}
For the proof, the interested reader is referred to \cite{Jaramillo13}.
\end{INFOCOM}
Lemma \ref{optimality_result} tells us that our algorithm is asymptotically optimal.

The proofs in Lemmas \ref{queue_stability} and \ref{optimality_result} follow the techniques in \cite{Jaramillo11b}, which are similar to the ideas in \cite{Neely05}. Slightly different results can be derived using the methods in \cite{Stolyar05, Eryilmaz05}.

%%%%%%%%%%%%%%%%%%%%%%%%%%%%%%%%%%%%%%%%%%%%%%%%%%%%%%%%%%%%%%%%%%%%
% *** SOLUTION TO THE ADMISSION CONTROL PROBLEM ***
%%%%%%%%%%%%%%%%%%%%%%%%%%%%%%%%%%%%%%%%%%%%%%%%%%%%%%%%%%%%%%%%%%%%
\subsection{Solution to the Admission Control Problem}
\label{solution_admission_control}

So far, in Section \ref{optimization_framework} we presented a model to serve a subset $\mathcal{L} \subseteq \mathcal{D}$ of links that allowed us in Section \ref{scheduler_ra} to develop a suitable scheduler and resource allocator. We will now consider the following problem: given that all the links in the set $\mathcal{L}$ can be served with an assigned flow rate $\bar{x}$, and that link $l \in \mathcal{D} \setminus \mathcal{L}$ requests to be admitted with flow rate  $\bar{x}$, can we determine, without knowledge of the capacity region and without disturbing the service rates of the links in the set $\mathcal{L}$, if  link $l$ can be served?

To do that, we will first introduce the notation $\mathcal{L}^+$ to denote the set formed by the set of links in  $\mathcal{L}$ and the link that wants to be admitted. Thus, we have that $|\mathcal{L}^+| = L+1$, and by abusing notation we will write $l \in \mathcal{L}^+$ to mean $l \in \{1, 2, \ldots, L , L+1\}$, where the index $L+1$ is assigned to the new link.

We will prove that by using the following utility function $U_l(x_l)$ at link $l \in \mathcal{L}^+$
\begin{equation*}
	U_l(x_l)=
		\left\{
			\begin{array}{ll}
				u_lx_l & \text{if } x_l \leq \bar{x} \\
				u_l\bar{x} & \text{if } x_l > \bar{x},
			\end{array}
		\right.
\end{equation*}
where $u_l>0$ is a suitable constant, we can achieve the desired goal. Note that the utility function increases with $x_l$ up to $\bar{x}$, and after that there is no gain in increasing the flow rate. Also, note that to completely define the utility function we only need to specify the \emph{utility parameter} $u_l$ since $\bar{x}$ is fixed.

The analysis will assume that
\begin{equation}
\label{utility_function_L}
	U_l(x_l)=
		\left\{
			\begin{array}{ll}
				u x_l & \text{if } x_l \leq \bar{x} \\
				u \bar{x} & \text{if } x_l > \bar{x}
			\end{array}
		\right.
	\text{ for all } l \in \mathcal{L}
\end{equation}
and
\begin{equation}
\label{utility_function_L+}
	U_l(x_l)=
		\left\{
			\begin{array}{ll}
				u_n x_l & \text{if } x_l \leq \bar{x} \\
				u_n \bar{x} & \text{if } x_l > \bar{x}
			\end{array}
		\right.
	\text{ for } l = L+1,
\end{equation}
and we will compare \eqref{opt_problem_old} and the optimization problem
\begin{equation}
\label{opt_problem_new}
	\max_{ \mu \in \mathcal{C}(\mathcal{L}^+), \mathbf{x} } \sum_{l \in \mathcal{L}^+} U_l(x_l)
\end{equation}
subject to
\begin{align*}
	& 0 \leq x_l \leq \mu_l \text{ for all } l \in \mathcal{L}^+ \\
	& x_l = 0 \text{ for all } l \in \mathcal{D} \setminus \mathcal{L}^+,
\end{align*}
to find the conditions on $u$ and $u_n$ such that we can make an admission control decision. We will call the optimization problem \eqref{opt_problem_old} over the set $\mathcal{L}$ the \emph{old system}, and the optimization \eqref{opt_problem_new} over the set $\mathcal{L}^+$ the \emph{new system}.

\begin{theorem}
\label{condition_cc1}
	If the utility functions are given by \eqref{utility_function_L}, \eqref{utility_function_L+}, and assuming the old system \eqref{opt_problem_old} can assign a flow rate of $\bar{x}$ to all links in $\mathcal{L}$, and if
\begin{equation*}
u_n < u \min_{l \in \mathcal{L}^+} \left\{ \bar{c}_l \right\} / \bar{c}_{L+1},
\end{equation*}
then the new system \eqref{opt_problem_new} will assign a flow rate of $\bar{x}$ to all $l \in \mathcal{L}$ and a rate $\hat{x} \leq \bar{x}$ to link $L+1$, where $\hat{x}$ is the maximum rate that can be assigned to link $L+1$ that allows to assign $\bar{x}$ to all other links.
$\hfill \diamond$
\end{theorem}
\begin{proof}
If there exists $\mu \in \mathcal{C}(\mathcal{L}^+)$ such that $\mu_l = \bar{x}$ for all $l \in \mathcal{L}^+$, we are done since  $(\mu^*, \mathbf{x}^*) =  (\mu, \{\bar{x}\}_{l \in \mathcal{L}^+})$ is a solution to the problem of the new system \eqref{opt_problem_new}. Thus, we will assume that link $L+1$ interferes with the set $\mathcal{I} \subseteq \mathcal{L}$ such that if $\mu_{L+1} > \hat{x}$, then some link in $\mathcal{I}$ must get an assigned flow rate strictly less than $\bar{x}$. Without loss of generality, we will proceed to do the analysis for link $L$ assuming that it is in $\mathcal{I}$.

Define the scheduling policy $P_{\mathcal{L}^+}^1(\mathbf{s})$ such that $x^1_{L+1} = \hat{x} = \sum_{\mathbf{s} \in \mathcal{S}(\mathcal{L^+})} \bar{c}_{L+1} s_{L+1} P_{\mathcal{L^+}}^1(\mathbf{s})$ and $x^1_l = \bar{x} = \sum_{\mathbf{s} \in \mathcal{S}(\mathcal{L^+})} \bar{c}_l s_l P_{\mathcal{L^+}}^1(\mathbf{s})$ for all $l \in \mathcal{L}$, where $\hat{x} < \bar{x}$ is the maximum rate that can be assigned to link $L+1$ that allows to assign $\bar{x}$ to all other links in $\mathcal{I}$.

Now consider schedules $\mathbf{s^1}, \mathbf{s^2} \in \mathcal{S}(\mathcal{L}^+)$ such that $P_{\mathcal{L}^+}^1(\mathbf{s^1}) >0$, $s^1_{L} = 1$, $s^1_{L+1} = 0$,  $s^2_{L} = 0$, $s^2_{L+1} = 1$. We know that they exist since link $L$ is in $\mathcal{I}$ and from the definition of $\hat{x}$. For small enough $\delta >0$, define the policy $P_{\mathcal{L}^+}^2(\mathbf{s})$ as follows
\begin{equation*}
	P_{\mathcal{L}^+}^2(\mathbf{s}) =
	\left\{
		\begin{array}{ll}
			P_{\mathcal{L}^+}^1(\mathbf{s}) & \text{ for } \mathbf{s} \neq \mathbf{s^1}, \mathbf{s^2} \\
			P_{\mathcal{L}^+}^1(\mathbf{s}) - \delta & \text{ for } \mathbf{s} = \mathbf{s^1} \\
			P_{\mathcal{L}^+}^1(\mathbf{s}) + \delta & \text{ for } \mathbf{s} = \mathbf{s^2}.
		\end{array}
	\right.
\end{equation*}

In this case, from the definition of $\mathcal{I}$ and since $x^2_l = \sum_{\mathbf{s} \in \mathcal{S}(\mathcal{L^+})} \bar{c}_l s_l P_{\mathcal{L^+}}^2(\mathbf{s})$ for all $l \in \mathcal{L}^+$, we have the following rate allocation
\begin{equation*}
	x^2_{l} =
	\left\{
	\begin{array}{ll}
		\bar{x} - \delta \bar{c}_L & \text{ if } l = L \\
		\hat{x} + \delta  \bar{c}_{L+1}  & \text{ if } l = L+1.
	\end{array}
	\right.
\end{equation*}

Comparing the objective function for both policies we have
\begin{align*}
	& \sum_{l \in \mathcal{L}^+} U_l(x^1_l) - \sum_{l \in \mathcal{L}^+} U_l(x^2_l) \\
	& \geq \left[ \sum_{l \in \mathcal{L}} u \bar{x} + u_n \hat{x} \right] \\
	& \qquad {} - \left[ \sum_{l=1}^{L-1} u \bar{x} + u \left( \bar{x} - \delta \bar{c}_L \right) + u_n \left( \hat{x} + \delta  \bar{c}_{L+1} \right) \right] \\
	& = u\delta \bar{c}_L - u_n \delta  \bar{c}_{L+1} \\
	& \geq u\delta \min_{l \in \mathcal{L}^+} \left\{ \bar{c}_l \right\} - u_n \delta  \bar{c}_{L+1} \\
	& = \delta \left( u \min_{l \in \mathcal{L}^+} \left\{ \bar{c}_l \right\} - u_n  \bar{c}_{L+1} \right).
\end{align*}

Since the analysis is valid for any link in $\mathcal{I}$, if $u_n < u \min_{l \in \mathcal{L}^+} \left\{ \bar{c}_l \right\} / \bar{c}_{L+1}$, then it is optimal to allocate a rate of $\bar{x}$ to all links in $\mathcal{L}$ and $\hat{x}$ to link $L+1$, which proves the theorem.
%\hfill $\blacksquare$
\end{proof}

It is interesting to note that Theorem \ref{condition_cc1} gives us the conditions such that the resource allocator \eqref{congestion_controller} can be used to test if there are enough resources to fulfill an admission request without disturbing other links. Since the analysis was done for a static optimization problem, we will now use Lemma \ref{optimality_result} to show how the choice of $u$ and $u_n$ influence the mean assigned rate over the actual network problem, which is dynamic and stochastic in nature.

\begin{theorem}
\label{thm_lower_bound_flows}
	If the utility functions are given by \eqref{utility_function_L}, \eqref{utility_function_L+}, and assuming the old system \eqref{opt_problem_old} can assign a flow rate of $\bar{x}$ to all links in $\mathcal{L}$, and if
\begin{equation*}
	u_n < u \min_{l \in \mathcal{L}^+} \left\{ \bar{c}_l \right\} / \bar{c}_{L+1},
\end{equation*}
then for $l \in \mathcal{L}$
\begin{equation}
\label{lower_bound_flows1}
	\lim_{T \rightarrow \infty} E\left[ \frac{1}{T} \sum_{t=1}^T  x_l(t) \right] \geq \bar{x} - \frac{B_6 \epsilon}{u} - \frac{u_n}{u} ( \bar{x} - \hat{x} ),
\end{equation}
and for $l = L+1$
\begin{equation}
\label{lower_bound_flows2}
	\lim_{T \rightarrow \infty} E\left[ \frac{1}{T} \sum_{t=1}^T  x_{L+1}(t) \right] \geq \hat{x} - \frac{B_6 \epsilon}{u_n}
\end{equation}
for some $B_6>0$, where $\mathbf{x}(t)$ is a solution to \eqref{congestion_controller}, and $\hat{x} \leq \bar{x}$ is the maximum rate that can be assigned to link $L+1$ that allows to assign $\bar{x}$ to all other links.
$\hfill \diamond$
\end{theorem}
\begin{proof}
Before we proceed, if the utility functions are given by \eqref{utility_function_L} and \eqref{utility_function_L+}, we note that \eqref{congestion_controller} can be rewritten for all $l \in \mathcal{L}^+$ as
\begin{equation*}
	x_l(t) \in \mathop{\arg \max}_{ 0 \leq x_l \leq \bar{x} } \frac{1}{\epsilon} U_l(x_l) - q_l(t) x_l.
\end{equation*}
Thus, for all $t$ we have that $x_l(t) \leq \bar{x}$. Also, Theorem \ref{condition_cc1} tells us that $\mathbf{x}^* =  \mu^* = (\{\bar{x}\}_{l \in \mathcal{L}}, \hat{x})$, where $(\mu^*, \mathbf{x}^*)$ is a solution to \eqref{opt_problem_new}.

Note that Lemma \ref{optimality_result} is valid for any $\mathcal{L} \in \mathcal{D}$, so it is also valid for $\mathcal{L}^+$. Thus, rewriting \eqref{optimality_equation} for $\mathcal{L}^+$, and using \eqref{utility_function_L}, \eqref{utility_function_L+}, Theorem \ref{condition_cc1}, and the fact that for all $t$ and $l$, $x_l(t) \leq \bar{x}$, we get
	\begin{align*}
		& \lim_{T \rightarrow \infty} \sum_{l \in \mathcal{L}^+} \left\{ U_l( x_l^* ) - U_l\left( E\left[ \frac{1}{T} \sum_{t=1}^T  x_l(t) \right] \right) \right\} \\
		& = \lim_{T \rightarrow \infty} \left( \sum_{l \in \mathcal{L}} \left\{ u \bar{x} - u E\left[ \frac{1}{T} \sum_{t=1}^T  x_l(t) \right] \right\} \right. \\
	& \left. \qquad {} + u_n \hat{x} - u_n E\left[ \frac{1}{T} \sum_{t=1}^T  x_{L+1}(t) \right] \right) \\
		& = u \bar{x} L + u_n \hat{x} - u_n \tilde{x}_{L+1} - u \sum_{l \in \mathcal{L}} \tilde{x}_l \\
		& \leq B_6 \epsilon,
	\end{align*}
where $B_6>0$ is a constant that is similarly found as $B_1$ for the case that  Lemma \ref{optimality_result} is rewritten for the set $\mathcal{L}^+$, and where we define
	\begin{equation*}
		\tilde{x}_l = \lim_{T \rightarrow \infty} E\left[ \frac{1}{T} \sum_{t=1}^T  x_l(t) \right] \text{ for all } l \in \mathcal{L}^+.
	\end{equation*}
Hence, we get the following inequality
\begin{equation}
\label{objective_bound}
	u \bar{x} L + u_n \hat{x} - u_n \tilde{x}_{L+1} - u \sum_{l \in \mathcal{L}} \tilde{x}_l \leq B_6 \epsilon.
\end{equation}

	Consider link $L+1$. Since $x_l(t) \leq \bar{x}$ for all $t$ and $l \in \mathcal{L}^+$, note that $\tilde{x}_l \leq \bar{x}$. Then
\begin{align*}
	\tilde{x}_{L+1} & \geq \frac{u}{u_n} \bar{x} L + \hat{x} - \frac{u}{u_n} \sum_{l \in \mathcal{L}} \tilde{x}_l - \frac{B_6 \epsilon}{u_n} \\
	& \geq \hat{x} - \frac{B_6 \epsilon}{u_n}.
\end{align*}
Now consider link $l \in \mathcal{L}$. From the fact that $\tilde{x}_l \leq \bar{x}$, we have that
\begin{align*}
	\tilde{x}_l & \geq \bar{x} L - \sum_{i \in \mathcal{L} \setminus \{ l \} } \tilde{x}_i + \frac{u_n}{u} \hat{x} - \frac{u_n}{u} \tilde{x}_{L+1} - \frac{B_6 \epsilon}{u} \\
	& \geq \bar{x} - \frac{B_6 \epsilon}{u} - \frac{u_n}{u} ( \bar{x} - \hat{x} ),
\end{align*}
which completes the proof.
%\hfill $\blacksquare$
\end{proof}

Theorem \ref{thm_lower_bound_flows} gives us a lower bound on the mean assigned rates, assuming that only one link will be responsible for the loss in performance of the stochastic system compared to the static optimization problem. For fixed $\epsilon$, equation \eqref{lower_bound_flows1} suggests that $u$ should be large to guarantee that the links in the set $\mathcal{L}$ get a mean assigned rate close to $\bar{x}$. Although \eqref{lower_bound_flows2} also suggests that $u_n$ should be large to make sure we accurately determine the maximum mean rate that we can assign to link $L+1$, if our goal is to disturb the mean assigned rates for links in $\mathcal{L}$ as little as possible, the ratio $u_n/u$ should be as small as possible, as we highlight in the following corollary.

\begin{corollary}
	If the utility functions are given by \eqref{utility_function_L}, \eqref{utility_function_L+}, and assuming the old system \eqref{opt_problem_old} can assign a flow rate of $\bar{x}$ to all links in $\mathcal{L}$, and if
\begin{equation*}
	u_n < u \min_{l \in \mathcal{L}^+} \left\{ \bar{c}_l \right\} / \bar{c}_{L+1},
\end{equation*}
then
\begin{equation*}
	\bar{x} L - \sum_{l \in \mathcal{L}} \lim_{T \rightarrow \infty} E\left[ \frac{1}{T} \sum_{t=1}^T  x_l(t) \right] \leq \frac{B_6 \epsilon}{u} + \frac{u_n}{u} \left( \tilde{x}_{L+1} - \hat{x} \right)
\end{equation*}
for some $B_6>0$.
$\hfill \diamond$
\end{corollary}
\begin{proof}
	We get the desired result by rewriting \eqref{objective_bound}.
%	Rewriting \eqref{objective_bound} we get
%\begin{align*}
%	\bar{x} L - \sum_{l \in \mathcal{L}} \tilde{x}_l & \leq \frac{B_6 \epsilon}{u} + \frac{u_n}{u} \left( \tilde{x}_{L+1} - \hat{x} \right) \\
%	& \leq \frac{B_6 \epsilon}{u} + \frac{u_n}{u} \left( \bar{x} - \hat{x} \right),
%\end{align*}
%where we used the fact that $\tilde{x}_l \leq \bar{x}$ for all $l \in \mathcal{L}^+$ as proved in Theorem \ref{thm_lower_bound_flows}.
%\hfill $\blacksquare$
\end{proof}

%%%%%%%%%%%%%%%%%%%%%%%%%%%%%%%%%%%%%%%%%%%%%%%%%%%%%%%%%%%%%%%%%%%%
%
% *** CAPACITY REGION ***
%
%%%%%%%%%%%%%%%%%%%%%%%%%%%%%%%%%%%%%%%%%%%%%%%%%%%%%%%%%%%%%%%%%%%%
\section{Capacity Region}
\label{capacity_region}

To design a suitable admission controller, we must prove that we have properly defined the capacity region. Hence, in Lemma \ref{instability} we will first show that if we allocate mean rates that are not in the capacity region, then there is no scheduling algorithm that can keep the queues stable. Second, in Lemma \ref{stability} we will show that if we allocate rates that are an interior point of the capacity region, there is an algorithm that can keep the queues stable.

\begin{lemma}
\label{instability}
	If $\mathbf{x} \notin \mathcal{C}(\mathcal{L})$, then no scheduling algorithm can keep the queues stable when the mean assigned flow rates are given by $\mathbf{x}$.
$\hfill \diamond$
\end{lemma}

The proof can be found in \begin{ARXIV} Appendix \ref{proof_instability}\end{ARXIV}\begin{INFOCOM} \cite{Jaramillo13}\end{INFOCOM}, and follows a technique similar to \cite{Eryilmaz05b}.

\begin{lemma}
\label{stability}
	 If $\mathbf{x} \in \mathcal{C}(\mathcal{L})/(1+\Delta)$ for some $\Delta > 0$ such that $x_l > 0$ for all $l \in \mathcal{L}$, then there exists a scheduler that keeps the queues stable when the mean assigned flow rates are $\mathbf{x}$.
$\hfill \diamond$
\end{lemma}

In \begin{ARXIV} Appendix \ref{proof_stability} \end{ARXIV} \begin{INFOCOM} \cite{Jaramillo13} \end{INFOCOM} the proof is presented. It follows a technique similar to the ideas presented in \cite{Neely05}.

%%%%%%%%%%%%%%%%%%%%%%%%%%%%%%%%%%%%%%%%%%%%%%%%%%%%%%%%%%%%%%%%%%%%
%
% *** CONCLUSIONS ***
%
%%%%%%%%%%%%%%%%%%%%%%%%%%%%%%%%%%%%%%%%%%%%%%%%%%%%%%%%%%%%%%%%%%%%
\section{Conclusions}

We have considered the problem of distributed admission control without knowledge of the capacity region in single-hop wireless networks for flows that require a pre-specified bandwidth from the network. To achieve this goal, we presented a utility maximization framework that allowed us to develop a scheduler and a distributed resource allocator. By properly choosing the utility function used by the resource allocator, we have proved that existing flows can be served with a pre-specified bandwidth, while the link requesting admission can determine the largest mean flow rate that can be assigned that avoids interfering with the service to other nodes.

%%%%%%%%%%%%%%%%%%%%%%%%%%%%%%%%%%%%%%%%%%%%%%%%%%%%%%%%%%%%%%%%%%%%
%
% *** APPENDIX ***
%
%%%%%%%%%%%%%%%%%%%%%%%%%%%%%%%%%%%%%%%%%%%%%%%%%%%%%%%%%%%%%%%%%%%%
\begin{ARXIV}
\appendices
%%%%%%%%%%%%%%%%%%%%%%%%%%%%%%%%%%%%%%%%%%%%%%%%%%%%%%%%%%%%%%%%%%%%
% *** PROOF OF QUEUES LEMMA ***
%%%%%%%%%%%%%%%%%%%%%%%%%%%%%%%%%%%%%%%%%%%%%%%%%%%%%%%%%%%%%%%%%%%%
\section{Proof of Lemma \ref{queue_stability}}
\label{proof_queue_stability}

We will start by first proving some auxiliary lemmas.

\begin{lemma}
\label{optimization_equivalence}
	The following optimization problems have the same solution
	\begin{equation*}
		\max_{ \mathbf{s} \in \mathcal{S}(\mathcal{L}) } \sum_{l \in \mathcal{L}} q_l(t) \bar{c}_l s_l = \max_{ \mu \in \mathcal{C}(\mathcal{L}) } \sum_{l \in \mathcal{L}} q_l(t) \mu_l.
	\end{equation*}
$\hfill \diamond$
\end{lemma}
\begin{proof}
	From Definition \ref{def_av_rates} we know that the following two optimization problems are equivalent
	\begin{equation*}
		\max_{ \mathbf{s} \in \mathcal{S}(\mathcal{L}) } \sum_{l \in \mathcal{L}} q_l(t) \bar{c}_l s_l = \max_{ \gamma \in  \Gamma(\mathcal{L}) } \sum_{l \in \mathcal{L}} q_l(t) \gamma_l.
	\end{equation*}

	Since the objective function is linear, the optimal value will not change if we perform the optimization problem over the convex hull of $\Gamma(\mathcal{L})$, which from Definition \ref{def_capacity} we know that it is $\mathcal{C}(\mathcal{L})$. Thus, the following two optimization problems have the same solution
	\begin{equation*}
		\max_{ \gamma \in  \Gamma(\mathcal{L}) } \sum_{l \in \mathcal{L}} q_l(t) \gamma_l = \max_{ \mu \in \mathcal{C}(\mathcal{L}) } \sum_{l \in \mathcal{L}} q_l(t) \mu_l,
	\end{equation*}
which proves the lemma.
%\hfill $\blacksquare$
\end{proof}

\begin{lemma}
\label{expected_drift}
	Consider the Lyapunov function $V(\mathbf{q}) = \frac{1}{2} \sum_{l \in \mathcal{L}} q_l^2$, and assume that $X_{max} \geq \max_{\mu \in \mathcal{C}(\mathcal{L})} \left\{ \max_{l \in \mathcal{L}} \{ \mu_l \} \right\}$. If there exists $\mu(\Delta) \in \mathcal{C}(\mathcal{L})/(1+\Delta)$ for some $\Delta > 0$ such that $\mu_l(\Delta) > 0$ for all $l \in \mathcal{L}$, then
\begin{align*}
	& E\left[ V(\mathbf{q}(t+1)) | \mathbf{q}(t) = \mathbf{q} \right] - V(\mathbf{q}) \\
	& \leq B_1 -  B_2 \sum_{l \in \mathcal{L}} q_l - \frac{1}{\epsilon} \sum_{l \in \mathcal{L}} \left[ U_l(\mu_l(\Delta)) - U_l(x_l(t)) \right]
\end{align*}
given $\epsilon > 0$, for some $B_1>0$, $B_2>0$, where $\mathcal{C}(\mathcal{L})/(1+\Delta)$ is a scaled version of $\mathcal{C}(\mathcal{L})$, and $\mathbf{x}(t)$ is a solution to \eqref{congestion_controller}.
$\hfill \diamond$
\end{lemma}
\begin{proof}
\begin{align}
	& E\left[ V(\mathbf{q}(t+1)) | \mathbf{q}(t) = \mathbf{q} \right] - V(\mathbf{q}) \notag \\
	& = E\left[ \left. \frac{1}{2} \sum_{l \in \mathcal{L}} q_l^2(t+1) \right\vert \mathbf{q}(t) = \mathbf{q} \right] - \frac{1}{2} \sum_{l \in \mathcal{L}} q_l^2 \notag \\
	& = E\left[ \frac{1}{2} \sum_{l \in \mathcal{L}} \left\{ \left[ q_l + a_l(t) - d_l(t) \right]^+ \right\}^2 \right] - \frac{1}{2} \sum_{l \in \mathcal{L}} q_l^2 \notag \\
	& \leq E\left[ \frac{1}{2} \sum_{l \in \mathcal{L}} \left[ q_l + a_l(t) - d_l(t) \right]^2 \right] - \frac{1}{2} \sum_{l \in \mathcal{L}} q_l^2 \notag \\
	& = E\left[ \sum_{l \in \mathcal{L}} q_l \left[ a_l(t) - d_l(t) \right] + \frac{1}{2} \sum_{l \in \mathcal{L}} \left[ a_l(t) - d_l(t) \right]^2 \right]  \notag \\
	& \leq E\left[ \sum_{l \in \mathcal{L}} q_l \left[ a_l(t) - d_l(t) \right] \right] + \frac{1}{2} \sum_{l \in \mathcal{L}} E\left[ a_l^2(t) + d_l^2(t) \right] \notag \\
	& \leq E\left[ \sum_{l \in \mathcal{L}} q_l \left[ a_l(t) - d_l(t) \right] \right] + B_1 \notag \\
	& = E\left[ \sum_{l \in \mathcal{L}} q_l \left[ x_l(t) -  \bar{c}_l s_l(t) \right] \right] + B_1 \label{mu_l} \\
	& = E\left[ \sum_{l \in \mathcal{L}} -\frac{1}{\epsilon} U_l(x_l(t)) + q_l x_l(t) + \frac{1}{\epsilon} U_l(x_l(t)) -  q_l \bar{c}_l s_l(t) \right] \notag \\
	& \qquad {} + B_1 \notag \\
	& \leq E\left[ \sum_{l \in \mathcal{L}} \left\{ -\frac{1}{\epsilon} U_l(\mu_l(\Delta)) + q_l \mu_l(\Delta) \right. \right. \notag \\
	& \qquad {} \left. \left. + \frac{1}{\epsilon} U_l(x_l(t)) -  q_l (1+\Delta) \mu_l(\Delta) \right\} \right] + B_1 \label{mu_delta} \\
	& = B_1 + \sum_{l \in \mathcal{L}} \left\{ -\frac{1}{\epsilon} U_l(\mu_l(\Delta)) + \frac{1}{\epsilon} U_l(x_l(t)) -  q_l \Delta \mu_l(\Delta) \right\} \notag \\
	& = B_1 - \sum_{l \in \mathcal{L}} q_l \Delta \mu_l(\Delta) - \frac{1}{\epsilon} \sum_{l \in \mathcal{L}} \left[ U_l(\mu_l(\Delta)) - U_l(x_l(t)) \right] \notag \\
	& \leq B_1 - B_2 \sum_{l \in \mathcal{L}} q_l - \frac{1}{\epsilon} \sum_{l \in \mathcal{L}} \left[ U_l(\mu_l(\Delta)) - U_l(x_l(t)) \right] \notag 
\end{align}
where $B_1 =L \left( X_{max}^2 + \sigma^2 + 1 \right)/2$, \eqref{mu_l} is a consequence of \eqref{def_dl} and the definition of $\mu_l(t)$ in \eqref{def_mu_l} and \eqref{simplified_mu_l}, \eqref{mu_delta} follows from \eqref{scheduler}, \eqref{congestion_controller}, Lemma \ref{optimization_equivalence}, and the fact that  $(1+\Delta) \mu(\Delta) \in \mathcal{C}(\mathcal{L})$, and $B_2 = \Delta \min_{l \in \mathcal{L} } \left\{ \mu_l(\Delta) \right\}$.
%\hfill $\blacksquare$
\end{proof}

Since the last term in the right-hand side of the inequality can be upper-bounded, we have that the expected drift is negative but for a finite set of values of $\mathbf{q}(t)$, and therefore the Markov chain $\mathbf{q}(t)$ is positive recurrent. As a consequence of this fact we can now prove Lemma \ref{queue_stability}.

\begin{proof}[Proof of Lemma \ref{queue_stability}]
\begin{align}
	& E\left[ V(\mathbf{q}(t+1)) | \mathbf{q}(t) = \mathbf{q} \right] - V(\mathbf{q}) \notag \\
	& \leq B_1 - B_2 \sum_{l \in \mathcal{L}} q_l + \frac{1}{\epsilon} \sum_{l \in \mathcal{L}} \left[ U_l(x_l(t)) - U_l(\mu_l(\Delta)) \right] \notag \\
	& \leq B_1 - B_2 \sum_{l \in \mathcal{L}} q_l + \frac{1}{\epsilon} \sum_{l \in \mathcal{L}} \left[ \left| U_l(x_l(t)) \right| + \left| U_l(\mu_l(\Delta)) \right| \right] \notag \\
	& \leq B_1 - B_2 \sum_{l \in \mathcal{L}} q_l + \frac{1}{\epsilon} \sum_{l \in \mathcal{L}} 2 \max_{ 0 \leq x_l \leq X_{max}} \left| U_l(x_l) \right|. \notag
\end{align}
This expected drift is conditioned on the event $\{ \mathbf{q}(t) = \mathbf{q} \}$, so averaging over possible values of $\mathbf{q}(t)$ we get
\begin{align}
	& E\left[ V(\mathbf{q}(t+1)) \right] - E\left[ V(\mathbf{q}(t)) \right] \notag \\
	& \leq B_1 -  B_2 E\left[ \sum_{l \in \mathcal{L}} q_l(t) \right] + \frac{1}{\epsilon} \sum_{l \in \mathcal{L}} 2 \max_{ 0 \leq x_l \leq X_{max}} \left| U_l(x_l) \right|. \notag
\end{align}
Since the Markov chain $\mathbf{q}(t)$ is positive recurrent we know that for any initial distribution of $\mathbf{q}(1)$ the distribution when $t \rightarrow \infty$ is unique and equal to the equilibrium probability distribution. Thus
\begin{equation*}
	0 \leq B_1 - B_2 \lim_{t \rightarrow \infty} E\left[ \sum_{l \in \mathcal{L}} q_l(t) \right] + \frac{1}{\epsilon} \sum_{l \in \mathcal{L}} 2 \max_{ 0 \leq x_l \leq X_{max}} \left| U_l(x_l) \right|.
\end{equation*}
Reordering terms we get
\begin{align*}
	& \lim_{t \rightarrow \infty} E\left[ \sum_{l \in \mathcal{L}} q_l(t) \right] \\
	& \leq \frac{B_1}{B_2} + \frac{1}{\epsilon} \frac{\sum_{l \in \mathcal{L}} 2 \max_{ 0 \leq x_l \leq X_{max}} \left| U_l(x_l) \right|}{B_2} \\
	& = B_3 + \frac{1}{\epsilon} B_4.
\end{align*}
%\hfill $\blacksquare$
\end{proof}

%%%%%%%%%%%%%%%%%%%%%%%%%%%%%%%%%%%%%%%%%%%%%%%%%%%%%%%%%%%%%%%%%%%%
% *** PROOF OF OPTIMALITY LEMMA ***
%%%%%%%%%%%%%%%%%%%%%%%%%%%%%%%%%%%%%%%%%%%%%%%%%%%%%%%%%%%%%%%%%%%%
\section{Proof of Lemma \ref{optimality_result}}
\label{proof_optimality_result}

To prove that our scheduler and resource allocator optimally solve the stochastic network problem, we will use a slightly weaker version of Lemma \ref{expected_drift}.
\begin{lemma}
\label{weak_expected_drift}
	Consider the Lyapunov function $V(\mathbf{q}) = \frac{1}{2} \sum_{l \in \mathcal{L}} q_l^2$, and assume that $X_{max} \geq \max_{\mu \in \mathcal{C}(\mathcal{L})} \left\{ \max_{l \in \mathcal{L}} \{ \mu_l \} \right\}$. Then
\begin{align*}
	& E\left[ V(\mathbf{q}(t+1)) | \mathbf{q}(t) = \mathbf{q} \right] - V(\mathbf{q}) \\
	& \leq B_1 -  B_5 \sum_{l \in \mathcal{L}} q_l - \frac{1}{\epsilon} \sum_{l \in \mathcal{L}} \left[ U_l( x_l^* ) - U_l(x_l(t)) \right]
\end{align*}
given $\epsilon > 0$, for some $B_1>0$, $B_5 \geq 0$, where $\mathbf{x}^*$ a solution to \eqref{opt_problem_old}, and $\mathbf{x}(t)$ is a solution to \eqref{congestion_controller}.
$\hfill \diamond$
\end{lemma}
The proof is almost identical to the proof for Lemma \ref{expected_drift}, which can be found in Appendix \ref{proof_queue_stability}, and it is therefore omitted.

\begin{proof}[Proof of Lemma \ref{optimality_result}]
	Rearranging terms in Lemma \ref{weak_expected_drift} we get
\begin{align}
	& \frac{1}{\epsilon} \sum_{l \in \mathcal{L}} \left[ U_l( x_l^* ) - U_l(x_l(t)) \right] \notag \\
	& \leq B_1 -  B_5 \sum_{l \in \mathcal{L}} q_l - E\left[ V(\mathbf{q}(t+1)) | \mathbf{q}(t) = \mathbf{q} \right] + V(\mathbf{q}) \notag \\
	& \leq B_1 - E\left[ V(\mathbf{q}(t+1)) | \mathbf{q}(t) = \mathbf{q} \right] + V(\mathbf{q}). \notag
\end{align}

This inequality is conditioned on the event $\{ \mathbf{q}(t) = \mathbf{q} \}$, so averaging over possible values of $\mathbf{q}(t)$ we obtain
\begin{align}
	& \frac{1}{\epsilon} E\left[ \sum_{l \in \mathcal{L}} U_l( x_l^* ) - U_l(x_l(t)) \right] \notag \\
	& \leq B_1 - E\left[ V(\mathbf{q}(t+1)) \right] + E\left[ V(\mathbf{q}(t)) \right]. \notag
\end{align}

Adding terms from time slots $1$ to $T$ and dividing both sides of the inequality by $T$ we get
\begin{align}
	& \frac{1}{\epsilon} E\left[ \sum_{l \in \mathcal{L}} \left\{ U_l( x_l^* ) - \frac{1}{T} \sum_{t=1}^T U_l(x_l(t)) \right\} \right] \notag \\
	& \leq B_1 - \frac{1}{T} E\left[ V(\mathbf{q}(T+1)) \right] + \frac{1}{T} E\left[ V(\mathbf{q}(1)) \right] \notag \\
	& \leq B_1 + \frac{1}{T} E\left[ V(\mathbf{q}(1)) \right], \label{lyapunov_pos}
\end{align}
where \eqref{lyapunov_pos} follows from the fact that the Lyapunov function is non-negative. Taking the limit as $T \rightarrow \infty$ and assuming that $E\left[ V(\mathbf{q}(1)) \right] < \infty$ we obtain
\begin{align*}
	\lim_{T \rightarrow \infty} E\left[ \sum_{l \in \mathcal{L}} \left\{ U_l( x_l^* ) - \frac{1}{T} \sum_{t=1}^T U_l(x_l(t)) \right\} \right] & \leq B_1 \epsilon.
\end{align*}

Rearranging terms
\begin{align*}
	\lim_{T \rightarrow \infty} \sum_{l \in \mathcal{L}} \left\{ U_l( x_l^* ) - E\left[ \frac{1}{T} \sum_{t=1}^T U_l(x_l(t)) \right] \right\} \leq B_1 \epsilon.
\end{align*}

Finally, using Jensen's inequality \cite{Boyd04} and the fact that $U_l(\cdot)$ is a concave function, we have
\begin{align}
	& \lim_{T \rightarrow \infty} \sum_{l \in \mathcal{L}} \left\{ U_l( x_l^* ) - U_l\left( E\left[ \frac{1}{T} \sum_{t=1}^T  x_l(t) \right] \right) \right\} \notag \\
	& \leq \lim_{T \rightarrow \infty} \sum_{l \in \mathcal{L}} \left\{ U_l( x_l^* ) - E\left[ \frac{1}{T} \sum_{t=1}^T U_l(x_l(t)) \right] \right\} \notag \\
	& \leq B_1 \epsilon. \notag
\end{align}
%\hfill $\blacksquare$
\end{proof}

%%%%%%%%%%%%%%%%%%%%%%%%%%%%%%%%%%%%%%%%%%%%%%%%%%%%%%%%%%%%%%%%%%%%
% *** PROOF OF INSTABILITY LEMMA ***
%%%%%%%%%%%%%%%%%%%%%%%%%%%%%%%%%%%%%%%%%%%%%%%%%%%%%%%%%%%%%%%%%%%%
\section{Proof of Lemma \ref{instability}}
\label{proof_instability}

The proof is based on the \emph{Strict Separating Hyperplane Theorem} \cite[Appendix B.3]{Luenberger03}, which says that since $\mathcal{C}(\mathcal{L})$ is a convex set and if $\mathbf{x} \notin \mathcal{C}(\mathcal{L})$, then there exists a vector $\mathbf{b}$ such that
\begin{equation}
\label{separating_hyperplane}
\mathbf{b}^T \mathbf{x} \geq \max_{\mathbf{z} \in \mathcal{C}(\mathcal{L})} \mathbf{b}^T \mathbf{z} + \beta
\end{equation}
for some $\beta > 0$. We can let $\mathbf{b} = \mathbf{x} - \mathbf{x_0}$, where
\begin{equation}
\label{closest_point}
	\mathbf{x_0} = \mathop{ \arg \min }_{ \mathbf{z} \in \mathcal{C}(\mathcal{L}) } | \mathbf{x} - \mathbf{z} |.
\end{equation}

First, we highlight that $\mathbf{b}$ is a non-negative vector. To see this, assume that there exists $\tilde{l} \in \mathcal{L}$ such that $b_{\tilde{l}} = x_{\tilde{l}} - x_{0\tilde{l}} < 0$. Defining the vector $\mathbf{x'}$ such that $x'_l = x_{0l}$ for $l \in \mathcal{L} \setminus \{ \tilde{l} \}$ and $x'_{ \tilde{l} } = x_{ \tilde{l} }$, we note that $\mathbf{x'} \in \mathcal{C}(\mathcal{L})$ and it can be checked that
\begin{equation*}
	| \mathbf{x} - \mathbf{x'} | < | \mathbf{x} - \mathbf{x_0} |,
\end{equation*}
which contradicts \eqref{closest_point}. Second, since  $\mathbf{x} \notin \mathcal{C}(\mathcal{L})$ we note that $\mathbf{b}$ must have at least a positive element.

Now let us consider the arrival process $a_l(t)$ such that $E[ a_l(t) ] = x_l$ for all time slots $t$ and all $l \in \mathcal{L}$, and consider that arrivals are independent of the queue lengths. Defining the function
\begin{equation*}
	V(\mathbf{q}) = \mathbf{b}^T \mathbf{q}
\end{equation*}
we have the following drift analysis
\begin{align}
	& E\left[ \left. V(\mathbf{q}(t+1)) \right| \mathbf{q}(t) = \mathbf{q} \right] - V(\mathbf{q}) \notag \\
	& = E\left[ \left. \sum_{l \in \mathcal{L}} b_l q_l(t+1) \right| \mathbf{q}(t) = \mathbf{q} \right] - \sum_{l \in \mathcal{L}} b_l q_l \notag \\
	& = E\left[ \left. \sum_{l \in \mathcal{L}} b_l \left\{ q_l + a_l(t) - d_l(t) \right\}^+ \right| \mathbf{q}(t) = \mathbf{q} \right] - \sum_{l \in \mathcal{L}} b_l q_l \notag \\
	& \geq E\left[ \left. \sum_{l \in \mathcal{L}} b_l \left\{ q_l + a_l(t) - d_l(t) \right\} \right| \mathbf{q}(t) = \mathbf{q} \right] - \sum_{l \in \mathcal{L}} b_l q_l \notag \\
	& = E\left[ \left. \sum_{l \in \mathcal{L}} b_l \left\{ a_l(t) - d_l(t) \right\} \right| \mathbf{q}(t) = \mathbf{q} \right] \notag \\
	& = \sum_{l \in \mathcal{L}} b_l \left\{ x_l - E\left[ \left. c_l(t) s_l(t) \right| \mathbf{q}(t) = \mathbf{q} \right] \right\} \label{mu_l2} \\
	& = \sum_{l \in \mathcal{L}} b_l \left\{ x_l - \mu_l(t) \right\} \notag \\
	& = \mathbf{b}^T \mathbf{x} - \mathbf{b}^T \mathbf{\mu}(t) \notag \\
	& \geq \mathbf{b}^T \mathbf{x} - \max_{\mathbf{z} \in \mathcal{C}(\mathcal{L})} \mathbf{b}^T \mathbf{z} \label{mu_in_C} \\
	& \geq \beta, \label{beta}
\end{align}
where \eqref{mu_l2} uses \eqref{def_dl}, \eqref{mu_in_C} follows from Definition \ref{def_capacity}, where it can be checked that $\mu(t) \in \mathcal{C}(\mathcal{L})$, and \eqref{beta} follows from \eqref{separating_hyperplane}. But this implies that $\lim_{ t \rightarrow \infty } E\left[ V(\mathbf{q}(t)) \right] = \infty$, which proves that the system is not stable in the mean.
\hfill $\blacksquare$

%%%%%%%%%%%%%%%%%%%%%%%%%%%%%%%%%%%%%%%%%%%%%%%%%%%%%%%%%%%%%%%%%%%%
% *** PROOF OF STABILITY LEMMA ***
%%%%%%%%%%%%%%%%%%%%%%%%%%%%%%%%%%%%%%%%%%%%%%%%%%%%%%%%%%%%%%%%%%%%
\section{Proof of Lemma \ref{stability}}
\label{proof_stability}

Consider the Bernoulli arrival process $a_l(t)$ such that $E[ a_l(t) ] = x_l < X$, variance upper-bounded by $\sigma^2$, $P(a_l(t)=0)>0$ and $P(a_l(t)=1)>0$ for all time slots $t$ and all $l \in \mathcal{L}$.\footnote{The last two conditions guarantee that the Markov chain $\mathbf{q}(t)$ is irreducible and aperiodic and can be replaced by similar assumptions.} Furthermore, consider that the arrivals are independent between time slots.

Define the scheduling policy $P_{\mathcal{L}}(\mathbf{s})$ such that
\begin{equation}
\label{service_rate}
	(1+\Delta) x_l = \sum_{\mathbf{s} \in \mathcal{S}(\mathcal{L})} \bar{c}_l s_l P_{\mathcal{L}}(\mathbf{s})
\end{equation}
for all $l \in \mathcal{L}$. From Definition \ref{def_capacity} we know that such policy exists since $(1+\Delta) \mathbf{x} \in \mathcal{C}(\mathcal{L})$ because $\mathbf{x} \in \mathcal{C}(\mathcal{L})/(1+\Delta)$.

If we use the Lyapunov function $V(\mathbf{q}) = \frac{1}{2} \sum_{l \in \mathcal{L}} q_l^2$, then we have the following drift analysis

\begin{align}
	& E\left[ V(\mathbf{q}(t+1)) | \mathbf{q}(t) = \mathbf{q} \right] - V(\mathbf{q}) \notag \\
	& = E\left[ \left. \frac{1}{2} \sum_{l \in \mathcal{L}} q_l^2(t+1) \right\vert \mathbf{q}(t) = \mathbf{q} \right] - \frac{1}{2} \sum_{l \in \mathcal{L}} q_l^2 \notag \\
	& = E\left[ \frac{1}{2} \sum_{l \in \mathcal{L}} \left\{ \left[ q_l + a_l(t) - d_l(t) \right]^+ \right\}^2 \right] - \frac{1}{2} \sum_{l \in \mathcal{L}} q_l^2 \notag \\
	& \leq E\left[ \frac{1}{2} \sum_{l \in \mathcal{L}} \left[ q_l + a_l(t) - d_l(t) \right]^2 \right] - \frac{1}{2} \sum_{l \in \mathcal{L}} q_l^2 \notag \\
	& = E\left[ \sum_{l \in \mathcal{L}} q_l \left[ a_l(t) - d_l(t) \right] + \frac{1}{2} \sum_{l \in \mathcal{L}} \left[ a_l(t) - d_l(t) \right]^2 \right]  \notag \\
	& \leq E\left[ \sum_{l \in \mathcal{L}} q_l \left[ a_l(t) - d_l(t) \right] \right] + \frac{1}{2} \sum_{l \in \mathcal{L}} E\left[ a_l^2(t) + d_l^2(t) \right] \notag \\
	& \leq E\left[ \sum_{l \in \mathcal{L}} q_l \left[ a_l(t) - d_l(t) \right] \right] + \beta_1 \notag \\
	& = E\left[ \sum_{l \in \mathcal{L}} q_l \left[ a_l(t) -  \bar{c}_l s_l(t) \right] \right] + \beta_1 \label{c_mu_l1} \\
	& = \sum_{l \in \mathcal{L}} q_l \left[ x_l -  (1+\Delta) x_l \right] + \beta_1 \label{c_mu_l2} \\
	& = \beta_1 - \Delta \sum_{l \in \mathcal{L}} q_l x_l \notag \\
	& \leq \beta_1 - \beta_2 \sum_{l \in \mathcal{L}} q_l \notag
\end{align}
where $\beta_1 =L \left( X^2 + \sigma^2 + 1 \right)/2$, \eqref{c_mu_l1} is a consequence of \eqref{def_dl}, \eqref{c_mu_l2} follows from \eqref{service_rate}, and $\beta_2 = \Delta \min_{l \in \mathcal{L} } \left\{ x_l \right\}$.

Since $\beta_2>0$ because $x_l > 0$ for all $l \in \mathcal{L}$, we have that the expected drift is negative but for a finite set of values of $\mathbf{q}(t)$, thus the Markov chain $\mathbf{q}(t)$ is positive recurrent. A consequence of this is that the expected queue lengths are finite, which concludes the proof.
\hfill $\blacksquare$
\end{ARXIV}

%%%%%%%%%%%%%%%%%%%%%%%%%%%%%%%%%%%%%%%%%%%%%%%%%%%%%%%%%%%%%%%%%%%%
%
% *** ACKNOWLEDGMENT ***
%
%%%%%%%%%%%%%%%%%%%%%%%%%%%%%%%%%%%%%%%%%%%%%%%%%%%%%%%%%%%%%%%%%%%%
%\section*{Acknowledgment}
%
%Acknowledgment goes here.
%
%%%%%%%%%%%%%%%%%%%%%%%%%%%%%%%%%%%%%%%%%%%%%%%%%%%%%%%%%%%%%%%%%%%%
%
% *** BIBLIOGRAPHY ***
\IEEEtriggeratref{19}


%\bibliographystyle{IEEEtran}
%\bibliography{IEEEabrv,../bibtex/delay}

\bibliographystyle{IEEEtran}
\bibliography{IEEEfull,../bibtex/admission,../bibtex/ACR}
%%%%%%%%%%%%%%%%%%%%%%%%%%%%%%%%%%%%%%%%%%%%%%%%%%%%%%%%%%%%%%%%%%%%
%
% *** THE END ***
%
%%%%%%%%%%%%%%%%%%%%%%%%%%%%%%%%%%%%%%%%%%%%%%%%%%%%%%%%%%%%%%%%%%%%
\end{document}